\documentclass[10pt,a4paper]{article}
\usepackage{a4wide}
\usepackage{amsmath}
\usepackage{epsfig}
\usepackage{psfrag}
\usepackage{amssymb}
\usepackage{wrapfig}
\usepackage{algorithm}
\usepackage{algorithmic}
\usepackage{paralist}
\usepackage{soul}
\usepackage{color}

\newcommand{\ceil}[1]{\left\lceil #1 \right\rceil}

\newcommand{\donotshow}[1]{}

\newcommand{\ignore}[1]{}

\newcommand{\CC}{C\raisebox{.08ex}{\hbox{\tt ++}}}

\newcommand{\N}{\mathbb{N}}

\newlength{\setspacing}
\newcommand{\mbegin}{\{\ \ }
\newcommand{\mend}{\}}

\newlength{\mleftindent}
\newlength{\mindent}
\newlength{\mboxwidth}
\newcommand{\mincrement}{\addtolength{\mboxwidth}{-\mindent}}
\newcommand{\mdecrement}{\addtolength{\mboxwidth}{\mindent}}

\newlength{\preprogramskip}
\newlength{\postprogramskip}
\newlength{\mexpwidth}
\newlength{\mexpindent}
\newcommand{\indentafterkeyword}{\hspace*{0.5em}}

\newcommand{\mslifelse}[3]  
{\setlength{\mexpwidth}{\mboxwidth}%
\settowidth{\mexpindent}{{\bf if\indentafterkeyword}}%
\addtolength{\mexpwidth}{-\mexpindent}%
{\bf if\indentafterkeyword}\parbox[t]{\mexpwidth}{#1}\\
\mincrement \mbegin \parbox[t]{\mboxwidth}{#2 \mend} \mdecrement \\
{\bf else} \\
\mincrement \mbegin \parbox[t]{\mboxwidth}{#3}\\
\mend \mdecrement
}

{\vspace{-0.3em}\begin{itemize}
\setlength{\itemsep}{0.2\itemsep}%
\setlength{\parskip}{0.2\parskip}%
\setlength{\topsep}{0.0\topsep}%
}
{\end{itemize}}

{\begin{itemize}
\setlength{\itemsep}{0.2\itemsep}%
\setlength{\parskip}{0.2\parskip}%
\setlength{\topsep}{0.0\topsep}%
}
{\end{itemize}}

{\begin{itemize}
\setlength{\itemsep}{1.5\itemsep}%
\setlength{\parskip}{1.5\parskip}%
\setlength{\topsep}{0.0\topsep}%
}
{\end{itemize}}

{\begin{itemize}
\setlength{\itemsep}{1.5\itemsep}%
\setlength{\parskip}{1.5\parskip}%
\par \smallskip
}
{\end{itemize}}

{\begin{itemize}
\setlength{\itemsep}{1.5\itemsep}%
\setlength{\parskip}{1.5\parskip}%
\par \medskip
}
{\end{itemize}}

{\vspace{-0.3em}\begin{description}
\setlength{\itemsep}{0.2\itemsep}%
\setlength{\parskip}{0.2\parskip}%
\setlength{\topsep}{0.0\topsep}%
}
{\end{description}}

{\begin{description}
\setlength{\itemsep}{0.2\itemsep}%
\setlength{\parskip}{0.2\parskip}%
\setlength{\topsep}{0.0\topsep}%
}
{\end{description}}

{\vspace{-0.3em}\begin{enumerate}
\setlength{\itemsep}{0.2\itemsep}%
\setlength{\parskip}{0.2\parskip}%
\setlength{\topsep}{0.0\topsep}%
}
{\end{enumerate}}

{\begin{enumerate}
\setlength{\itemsep}{0.2\itemsep}%
\setlength{\parskip}{0.2\parskip}%
\setlength{\topsep}{0.0\topsep}%
}
{\end{enumerate}}

\newlength{\proofpostskipamount}\newlength{\proofpreskipamount}
\newenvironment{proof}%
               {\par\vspace{\proofpreskipamount}\noindent{\bf Proof:}\hspace{0.5em}}
               {\nopagebreak%
                \strut\nopagebreak%
                \hspace{\fill}\qed\par\vspace{\proofpostskipamount}\noindent}

               {\par\vspace{0.5ex}\noindent{\bf Proof #1:}\hspace{0.5em}}%
               {\nopagebreak%
                \strut\nopagebreak%
                \hspace{\fill}\qed\par\medskip\noindent}

\newtheorem{theorem}{Theorem}

\providecommand{\qed}{\rule[-0.2ex]{0.3em}{1.4ex}}

\newlength{\mydefwidth}
\newlength{\mytextwidth}

\newcommand{\myurl}[1]{{\footnotesize \url{#1}}}

\usepackage{mathptmx}
\usepackage[scaled]{helvet}   
\usepackage{courier}

\usepackage[T1]{fontenc}

\usepackage{rotating}
\usepackage{adjustbox}
\usepackage{url}

\DeclareUrlCommand\email{\urlstyle{rm}%
}

\newcommand{\barM}{\overline{M}}

\newcommand{\tildeM}{\widetilde{M}}
\newcommand{\tildeB}{\widetilde{B}}

\newcommand{\Pat}[1]{\textcolor{blue}{#1}}
\newcommand{\Pattwo}[2]{{\tiny \textcolor{green}{#1}}\textcolor{blue}{#2}}

\begin{document}

\title{Cache-Oblivious VAT-Algorithms}
\author{Tomasz Jurkiewicz\footnote{Google Z\"{u}rich, 
\protect\email{tomasz.tojot.jurkiewicz@gmail.com}
} \and Kurt Mehlhorn\thanks{MPI for Informatics, 
\protect\email{mehlhorn@mpi-inf.mpg.de}
} \and Patrick Nicholson\thanks{MPI for Informatics, \protect\email{pnichols@mpi-inf.mpg.de}
}}

\maketitle

\begin{abstract} The VAT-model (virtual address translation model) extends the EM-model (external memory model) and takes the cost of address translation in virtual memories into account. In this model, the cost of a single memory access may be logarithmic in the largest address used. We show that the VAT-cost of cache-oblivious algorithms is only a constant factor larger than their EM-cost; this requires a somewhat more stringent tall cache assumption than for the EM-model.  \end{abstract}

\section{Introduction} Modern processors have a memory hierarchy and use virtual memory. We concentrate on two-levels of the hierarchy and refer to the faster memory as the cache. Data is moved between the fast and the slow memory in blocks of contiguous memory cells, and only data residing in the fast memory can be accessed directly. Whenever data in the slow memory is accessed, a cache fault occurs and the block containing the data must be moved to the fast memory. In the EM-model of computation, the complexity of an algorithm is defined as the number of cache faults. 

In general, many processes are running concurrently on the same machine. Each running process has its own linear address space $0$, $1$, $2$, \ldots, and the operating system ensures that the distinct linear address spaces of the distinct processes are mapped injectively to the physical address space of the processor. To this effect, the operating system maintains a translation tree for each process. The translation from virtual addresses to physical addresses is implemented as a tree walk and incurs cost; see Section~\ref{sec: VAT-model} for details. The depth of the tree for a particular process is $d = \log_K (m/P)$ where $m$ is the maximum address used by the process, $K$ is the arity of the tree and $P$ is the page size. Typically, $K \approx P \approx 2^{10}$, and $K$ is chosen such that the space requirement of a translation tree node is equal to the size of a page. The translation process accesses $d$ pages and only pages residing in fast memory can be accessed directly. Any node visited during the translation process must be brought into fast memory if not already there, and hence a single memory access may cause up to $d$ cache faults. 

The cost of the translation process is clearly noticeable in some cases. Jurkiewicz and Mehlhorn~\cite{Cost-of-Address-Translation} timed some simple programs and observed that for some of them the quotient
\[        \text{measured running time for input size $n$}/\text{RAM-running time for input size $n$} \]
seems to grow logarithmically in $n$\Pattwo{,}{;} see Figure~\ref{fig: experiments}. Jurkiewicz and Mehlhorn introduced the VAT-model as an extension of the EM-model to account for the cost of address translation; see Section~\ref{sec: VAT-model} for a definition of their model. They showed that the growth rates of the  measured running times of the programs mentioned in Figure~\ref{fig: experiments} are correctly predicted by the model.

\begin{figure}[t]
  	\begin{center}
	\begin{tabular}{cc}
	\adjustbox{valign=m}{\begin{sideways}\mbox{running time/RAM complexity}\end{sideways}}&
	\adjustbox{valign=m}{\includegraphics[width=0.9\textwidth]{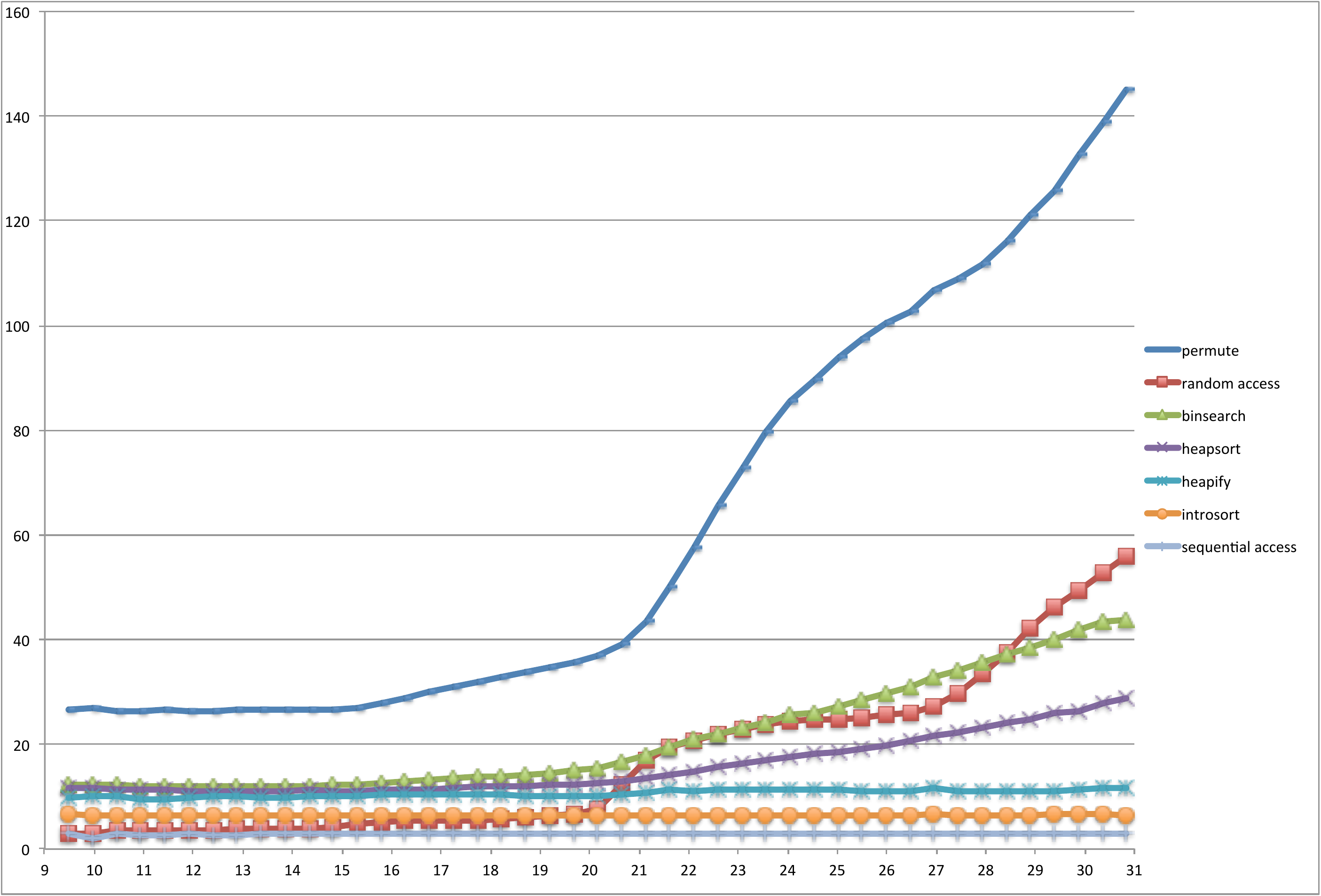}}\\
	&$\log(\text{input size})$
	\end{tabular}
  	\end{center}
\caption{
The abscissa shows the logarithm of the input size.
The ordinate shows the measured running time divided by the RAM-complexity (normalized operation time). RAM complexity of the programs shown is either $cn$ or $cn \log n$; \cite{Cost-of-Address-Translation} ignores lower order terms. The constants $c$ were chosen such that the different plots fit nicely into the same figure. 
The normalized operation times of sequential access, quicksort (introsort), and heapify are
constant, the normalized operation times of permute, random scan,  repeated binary search, heapsort grow as functions of the problem size. Note that a straight-line corresponds to a linear function in $\log(\text{input size})$. 
}\label{fig: experiments}
\end{figure}

The EM-model penalizes non-locality and the VAT-model penalizes it even more. Cache-oblivious algorithms~\cite{FrigoEtAl12} show good data locality for all memory sizes. Jurkiewicz and Mehlhorn~\cite{Cost-of-Address-Translation} showed that some cache-oblivious algorithms, namely those that do no need a tall cache assumption, also perform well in the VAT-model. Their paper poses the question of  whether a similar statement can be made for the much larger class of cache-oblivious algorithms that require a tall cache assumption. We answer their question in the affirmative. 

Our main result is as follows. Consider a cache-oblivious algorithm that incurs $C(\tildeM,\tildeB,n)$ cache faults, when run on a machine with cache size $\tildeM$ and block size $\tildeB$, provided that $\tildeM \ge g(\tildeB)$. Here $g: \N \mapsto \N$ is a function that captures the ``tallness'' requirement on the cache\Pat{~\cite{FrigoEtAl12}}. We consider the execution of the algorithm on a VAT-machine with cache size $\barM$ and page size $P$ and show that the number of cache faults is bounded by $4d C(M/4,dB,n)$ provided that $M \ge 4g(dB)$. Here $M = \barM/a$, $B = P/a$ and $a \ge 1$ is the size (in addressable units) of the items handled by the algorithm. 

Funnel sort~\cite{FrigoEtAl12} is an optimal cache-oblivious sorting algorithm. On an EM-machine with cache size $\tildeM$ and block size $\tildeB$, it  sorts $n$ items with 
 \[        C(\tildeM,\tildeB,n) = O\left(\frac{n}{\tildeB} \ceil{\frac{\log n/\tildeM}{\log \tildeM/\tildeB}}\right)  \]
cache faults provided that $\tildeM \ge \tildeB^2$: thus $g$ is quadratic\footnote{This constraint can be reduced to $\tildeM \ge \tildeB^{1+\varepsilon}$ for any constant $\varepsilon > 0$~\cite{BrodalFagerbergFunnelsort}, however we do not wish to introduce additional notation.} for this algorithm. As a consequence of our main theorem, we obtain:

\begin{theorem} Funnel sort sorts $n$ items, each of size $a \ge 1$ addressable units, on a VAT-machine with cache size $\barM$ and page size $P$, with at most 
\[  O\left(\frac{4n}{B} \ceil{\frac{\log 4n/M}{\log M/(4dB)}}\right)\]
cache faults, where $M = \barM/a$ and $B = P/a$. This assumes $(B \log_K(2n/P))^2 \le M/4$. \end{theorem}

Since $M/(4dB) \ge (M/B)^{1/2}$ for realistic values of $M$, $B$, $K$, and $n$, this implies funnel-sort is essentially optimal also in the VAT-model. 

\section{The EM-Model}\label{sec: EM-model}

An EM-machine has a fast memory of size $\tildeM$ and a data is moved between fast and slow memory in blocks of size $\tildeB$. Algorithms for EM-machines may use $\tildeM$ and $\tildeB$ in the program code; algorithms are 
not written for specific values of $\tildeM$ and $\tildeB$, but work for any values of $\tildeM$ and $\tildeB$ satisfying certain constraints, e.g., that the fast memory can hold a certain number of blocks. We capture these constraints by a function $g: \N \mapsto \N$ and the requirement $\tildeM \ge g(\tildeB)$. 

Cache-oblivious algorithms are algorithms that do not refer to the parameters $\tildeM$ and $\tildeB$ in the code. Only the analysis is done in terms of $\tildeM$ and $\tildeB$. Frequently, the analysis only holds for $\tildeM$ and $\tildeB$ satisfying certain constraints; e.g., that the cache is tall and satisfies $\tildeM \ge \tildeB^2$. Again, this can be captured by an appropriate function $g$. 

It is customary in the EM-literature that the size of the fast memory and the size of a block are expressed in terms of
number of items handled by the algorithm. For example, for a sorting algorithm the items are the objects to be sorted. 

\section{The VAT-Model~\cite[Section XXX]{Cost-of-Address-Translation}}\label{sec: VAT-model}

VAT-machines are EM-machines that use virtual addresses. We concentrate on the virtual memory of a single program. Both real (physical) and virtual addresses are strings in $\{0, K-1\}^d\{0,P-1\}$. Any such string corresponds to a number in the interval $[0,K^d P - 1]$ in a natural way.
The $\{0,K-1\}^d$ part of the address is called an \emph{index}, and its length $d$ is an execution parameter fixed prior to the execution. We assume $d=\lceil\log_K(\text{last used address}/P)\rceil$.
The $\{0,P-1\}$ part of the address is called page offset and $P$ is the page size. A page contains $P$ addressable units, usually bytes.\footnote{In actual systems $K$ is chosen such that a node fits exactly into a page. For example, for the 64-bit addressing mode of the processors of the AMD64 family (see http://en.wikipedia.org/wiki/X86-64), the addressable units are bytes and $P = 2^{12}$. Since an address consists of $2^3$ bytes, $K = 2^9$.}
The translation process is a tree walk. We have a $K$-ary tree $T$ of height $d$.
The nodes of the tree are pairs $(\ell,i)$ with $\ell \ge 0$ and $i \ge 0$.
We refer to $\ell$ as the layer of the node and to $i$ as the number of the node.
The leaves of the tree are on layer zero and a node $(\ell, i)$ on layer $\ell \ge 1$ has $K$ children on layer $\ell-1$, namely the nodes $(\ell-1,Ki+a)$, for $a=0\ldots K-1$.
In particular, node $(d,0)$, the root, has children $(d-1,0),\ \ldots,\ (d-1,K-1)$.
The leaves of the tree correspond to physical pages of the main memory of a RAM machine. In order to translate a virtual address $x_{d-1}\ldots x_0 y$, we start in the root of $T$, and then follow the path described by ${x_{d-1}}\ldots{x_0}$. We refer to this path as the \emph{translation path} for the address. 
The path ends in the leaf $(0,\sum_{0 \leqslant i \leqslant d-1} x_i K^i)$. Then the offset $y$ selects the $y$-th cell in this page. 

\begin{figure}[t]
\begin{center}
\includegraphics[width=0.9\textwidth]{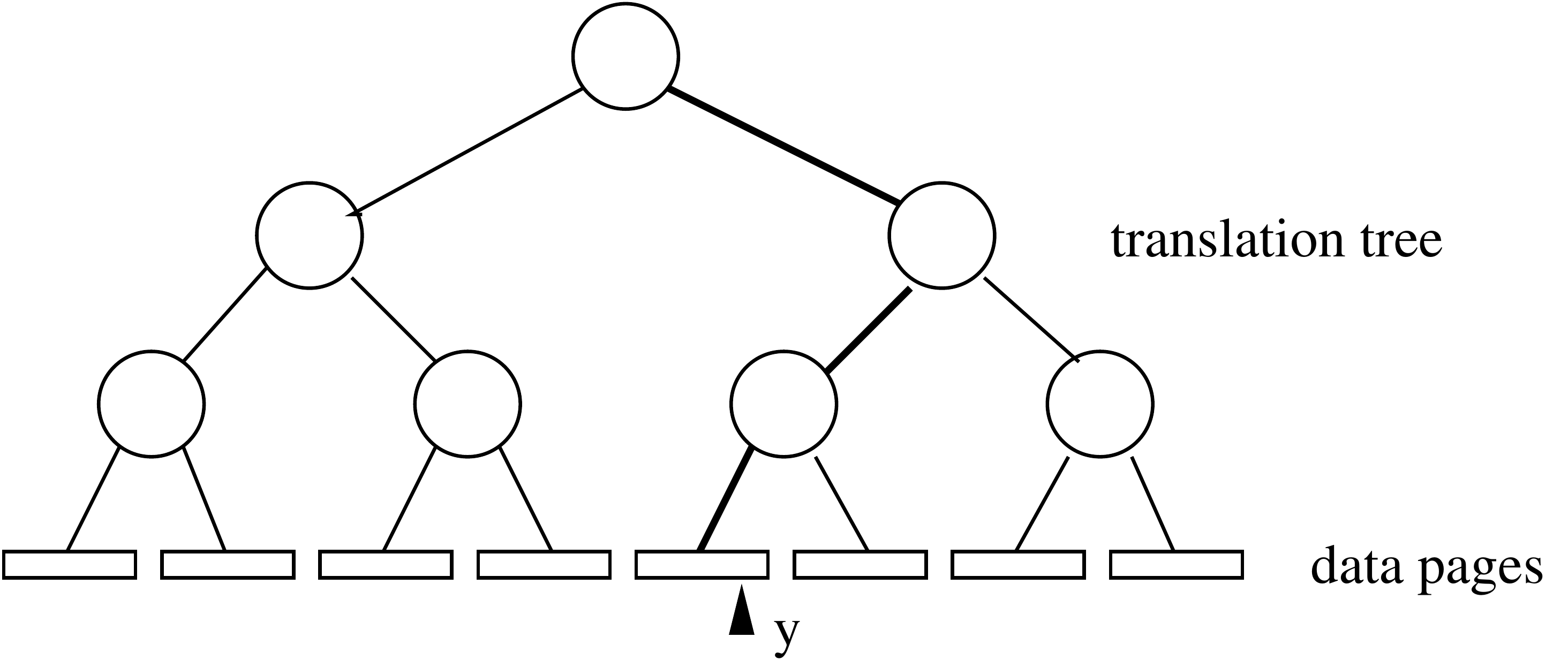}
\end{center}
\caption{\label{fig:translationtree} The pages holding the data are shown at the bottom and the translation tree is shown above the data pages. The translation tree has fan-out $K$ and depth $d$; here $K = 2$ and $d = 3$. The translation path for the virtual index 100 is shown. The offset $y$ selects a cell in the physical page with virtual index 100. The nodes of the translation tree and the data pages are stored in memory. Only nodes and data pages in fast memory (cache memory) can be accessed directly, nodes and data pages currently in slow memory have to be brought into fast memory before they can accessed. Each such move is a cache fault. In the EM-model only cache faults for data pages are counted, in the VAT-model, we count cache faults for nodes of the translation tree and for data pages. }
\end{figure}

A VAT-machine has a fast memory (cache) of size $\barM$ and data is moved between fast and slow memory in units of $P$ cells. We assume that a node of the translation tree fits into a page. Only nodes and data pages in the cache can be accessed directly, nodes and data pages currently in slow memory have to be brought into fast memory before they can be accessed. 
More precisely, let $a$ be a virtual address, and let $v_d,v_{d-1},\ldots,v_0$ be its translation path. Here, $v_d$ is the root of the translation tree, $v_d$ to $v_1$ are internal nodes of the translation tree, and $v_0$ is a data page. We assume that the root node is always in cache. Translating $a$ requires accessing all nodes of the translation path in order. Only nodes in the cache can be accessed. If a node needs to be accessed, but is not in the cache, it needs to be added to the cache, and some other page has to be evicted. The translation of $a$ ends when $v_0$ is accessed.
The cost of the memory access is the number of page faults incurred during the translation process. 

EM- and VAT-machines move cells in contiguous blocks. If the items handled by an EM-machine comprise $a \ge 1$ addressable units, we have $M = \barM/a$ and $B = P/a$. In the EM cost model only cache faults for data pages are charged and in the VAT cost model all cache faults arising in a memory access are charged.

\section{EM-Algorithms as VAT-Algorithms}

In the worst case, a memory access causes one cache fault in the EM-model and $d$ cache faults in the VAT-model. In order to amortize the cost of a cache fault in the EM-model, it suffices to access a significant fraction of the memory cells in each page brought to fast memory. Therefore, an EM-algorithm that is efficient for block size $d P/a$ should also be efficient in the VAT-model. The following discussion captures this intuition. 

For an EM-algorithm, let $C(\tildeM,\tildeB,n)$ be the number of IO-operations on an input of size $n$, where $\tildeM$ is the size of the faster memory (also called cache memory), $\tildeB$ is the block size, and $\tildeM \ge g(\tildeB)$.  The following theorem shows that any EM-algorithm that is aware of $\tildeM$ and $\tildeB$ implies the existence of an VAT-algorithm that is aware of $\barM$ and $P$.

\begin{theorem}\label{thm: EM-to-VAT}
Let $g:\N \mapsto \N$. Consider an EM-algorithm with IO-complexity $C(\tildeM,\tildeB, n)$, where $\tildeM$ is the size of the cache, $\tildeB$ is the size of a block, and $n$ is the input size, provided that $\tildeM \ge g(\tildeB)$.  Let $d = \log_K(n/P)$ and assume $g(dP/a) \le \barM/(4a)$, where $a \ge 1$ is the item size in number of addressable units. The program can be made to run on a VAT-machine with cache size $\barM$ and page size $P$ with at most $4 d C( \barM/(4a),      d P/a  ,n)$ cache faults. With the notation $M = \barM/a$ and $B = P/a$, the upper bound can be stated as $4d C(M/4, d B, n)$ and the tallness requirement becomes $M \ge 4 g(dB)$. 
\end{theorem}
\begin{proof} We run the EM-algorithm with a cache size of $\tildeM = \barM/(4a)$ and a block size $\tildeB = dP/a$ and show how to execute it efficiently in a VAT-machine with page size $P$ and cache size $\barM$. 
We use one-fourth of the cache for data and three-fourth for nodes of the translation tree. Since 
\[    \tildeM =  \barM/(4a) \ge g(dP/a) = g(\tildeB), \]
the number of cache faults incurred by the EM-algorithm is at most $C(\tildeM,\tildeB,n)$. 

Whenever, the EM-algorithm moves a block containing $\tildeB$ items) to its data cache, the VAT-machine moves the corresponding $d$ pages to the data cache and also moves all internal nodes of the translation paths to these $d$ pages to then translation cache. The number of internal nodes is bounded by $2d + \sum_{i \ge 1} d/K^i \le 2d + d/(K-1) \le 3d$. Thus a translation cache of size $3\barM/4$ suffices to store the translation paths to all pages in the data cache. For every cache fault of the EM-model, the VAT-machine incurs $4d$ cache faults. The theorem follows. \end{proof}

We apply the theorem to multi-way mergesort, an optimal sorting algorithm in the EM-model. It first creates $n/\tildeM$ sorted runs of size $\tildeM$ each by sorting chucks of size $\tildeM$ in internal memory. It then performs multi-way merge sort on these runs. For the merge, it keeps one block from each input run, one block of the output run, and a heap containing the first elements of each run in fast memory. If $\tildeM \ge 5 \tildeB$, we can merge two sequences since the space for the heap is certainly no more that the space for the input runs. The scheme results in a merge factor of $\Theta(\tildeM/\tildeB)$. We assume for simplicity that the factor is exactly $\tildeM/\tildeB$. Thus 
\[        C(\tildeM,\tildeB,n) = \frac{n}{\tildeB} \left( 1 + \ceil{\frac{\log n/\tildeM}{\log \tildeM/\tildeB}}\right).  \]

By Theorem~\ref{thm: EM-to-VAT} and with a cache size of $\barM$ and page size $P$, the number of cache faults in the VAT-model is at most 
\[   4 d C( \barM/(4a),      d P/a  ,n) = \frac{4n}{P/a} \left(1 + \ceil{\frac{\log 4an/\barM}{\log \barM/(4dP)}}\right) =  \frac{4n}{B} \left(1 + \ceil{\frac{\log 4n/M}{\log M/(4dB)}}\right) \]
Here $a$ is the item size in number of addressable units, $M = \barM/a$ and $B = P/a$. 
This assumes $5dP \le \barM/4$. If $M/(4dB) \ge (M/B)^{1/2}$, the asymptotic number of cache faults is the same in both models. For realistic values of $M$, $B$, and $n$, this will be the case. 

\ignore{\begin{theorem} In the VAT-model with cache size $M$ and page size $P$, $n$ items can be sorted with at most
\[  \frac{4n}{P} \left(1 + \ceil{\frac{\log 4n/M}{\log M/(4dP)}}\right)\]
cache faults. This assumes $5 P \log_K(2n/P) \le M/4$. \end{theorem}

Thus sorting in the VAT-model has essentially the same cost as sorting in the EM-model. }

\section{Cache-Oblivious VAT-Algorithms}

We next extend the theorem to cache-oblivious algorithms. Unlike the previous theorem, the following theorem indicates that \emph{any} algorithm, regardless of whether it is aware of the cache and block sizes, implies the existence of a VAT-algorithm that is similarly oblivious to $\barM$ and $P$.

\begin{theorem}\label{thm: cache-oblivious VAT} 
Let $g: \N \mapsto \N$. Let $A$ be an algorithm that incurs  $C(\tildeM,\tildeB,n)$ cache faults in the EM-model with cache size $\tildeM$ and block size $\tildeB$  on an input of size $n$,  provided that $\tildeM \ge g(\tildeB)$. Let $d = \log_K(n/P)$ and assume $g(dP/a) \le \barM/(4a)$, where $a \ge 1$ is the item size in number of addressable units. 
On a VAT-machine with cache size $\barM$ and page size $P$ and optimal use of the cache, the program incurs at most $4 d C( \barM/(4a),      d P/a  ,n)$.  With the notation $M = \barM/a$ and $B = P/a$, the upper bound can be stated as $4d C(M/4, d B, n)$ and the tallness requirement becomes $M/4 \ge g(dB)$. \end{theorem}
\begin{proof} We can almost literally reuse the proof of Theorem~\ref{thm: EM-to-VAT}. The optimal execution of $A$ on a VAT-machine with cache size $\barM$ and page size $P$ cannot incur more cache faults than the particular execution that we describe next. 

Consider an execution of $A$ on an EM-machine with cache size $\tildeM = \barM/(4a)$ and block size $\tildeB = dP/a$. Since 
\[    \tildeM =  \barM/(4a) \ge g(dP/a) = g(\tildeB), \]
the number of cache faults incurred by the EM-algorithm is at most $C(\tildeM,\tildeB,n)$. 
We execute the program on the VAT-machine as in the proof of Theorem~\ref{thm: EM-to-VAT}. 

We use one-fourth of the cache for data and three-fourth for nodes of the translation tree. Whenever, the EM-machine moves a block (of size $\tildeB$ items) to its data cache, the VAT-machine moves the corresponding $d$ pages to the data cache and also moves all internal nodes of the translation paths to these $d$ pages to then translation cache. The number of internal nodes is bounded by $2d + \sum_{i \ge 1} d/K^i \le 2d + d/(K-1) \le 3d$. Thus a translation cache of size $3\barM/4$ suffices to store the translation paths to all pages in the data cache. For every cache fault of the EM-machine, the VAT-machine incurs $4d$ cache faults. The theorem follows. \end{proof}

The optimal cache replacement strategy may be replaced by LRU at the cost of doubling the cache size and doubling the number of cache faults~\cite{Cost-of-Address-Translation}.

Recall that a cache-oblivious algorithm has no knowledge of the memory and the block size. It does well for any choice of $\tildeM$ and $\tildeB$ as long as $\tildeM \ge g(\tildeB)$. The theorem tells us that it also does well in the VAT-model as long as the more stringent requirement $M \ge 4a g(dB)$ is satisfied. 

\section{Conclusion}

We have shown that performance of cache-oblivious algorithms in the VAT-model matches their performance in the EM-model provided a somewhat more stringent tall cache assumption holds. 


\end{document}